\documentclass{llncs}
\usepackage{graphicx}
\usepackage{doc}
\usepackage{environ}

\usepackage{orcidlink}

\usepackage{amsmath,amssymb}
\newcommand{\C}{{C}}

\newcommand{\NP}{\textsf{NP}}
\newcommand{\PSPACE}{\textsf{PSPACE}}

\usepackage{bbm}
\usepackage{mathtools}
\usepackage{mathdots}
\usepackage{tikz}
\usetikzlibrary{arrows.meta}
\usetikzlibrary{calc,shapes.callouts,shapes.arrows}
\usetikzlibrary{shapes.misc}
\newcommand{\buyer}{B:\hspace{1.6mm}}
\newcommand{\seller}{S:\hspace{1.6mm}}

\hypersetup{
     colorlinks=true,
     linkcolor=blue,
     filecolor=blue,
     citecolor = blue,      
     urlcolor=cyan}

\usepackage[capitalise]{cleveref}
\newtheorem{contract}[theorem]{Mechanism}
\newtheorem{algorithm}[theorem]{Algorithm}

\crefname{contract}{Mechanism}{Mechanisms}
\makeatletter\@twosidefalse\@mparswitchfalse\makeatother

\newcommand{\agentL}{1}
\newcommand{\agentF}{2}
\usepackage{adjustbox}

\makeatletter
\newsavebox{\measure@tikzpicture}
\NewEnviron{scaletikzpicturetowidth}[1]{%
  \def\tikz@width{#1}%
  \def\tikzscale{1}\begin{lrbox}{\measure@tikzpicture}%
  \BODY
  \end{lrbox}%
  \pgfmathparse{#1/\wd\measure@tikzpicture}%
  \edef\tikzscale{\pgfmathresult}%
  \BODY
}
\makeatother
\begin{document}
\title{Which Games are Unaffected by \\Absolute Commitments?}
\subtitle{\small Full Version\thanks{This paper is the full version of an extended abstract with the same title that appears in The 23rd International Conference on Autonomous Agents and Multi-Agent Systems (AAMAS'24).}}
\author{Daji Landis \orcidlink{0000-0002-9985-0552} \and Nikolaj I. Schwartzbach \orcidlink{0000-0002-0610-4455} \thanks{NIS was supported by the European Research Council (ERC) under the European Union’s Horizon 2020 research and innovation programme (Grant agreement No. 101019547). Part of the work was done while NIS was a Ph.D. student at Aarhus University.}}
\institute{Bocconi University}

\maketitle \thispagestyle{empty}

\begin{abstract}
    We identify a subtle security issue that impacts mechanism design in scenarios in which agents can absolutely commit to strategies. Absolute commitments allow the strategy of an agent to depend on the commitments made by the other agents. This changes fundamental game-theoretic assumptions by inducing a meta-game in which agents choose which strategies they commit to. \\
    
    We say that a game that is unaffected by such commitments is \emph{Stackelberg resilient} and show that computing it is intractible in general, although it can be computed efficiently for two-player games of perfect information. We show the intuitive, but technically non-trivial result, that, if a game is resilient when some number of players have the capacity to make commitments, it is also resilient when these commitments are available to fewer players. We demonstrate the non-triviality of Stackelberg resilience by analyzing two escrow mechanisms from the literature. These mechanisms have the same intended functionality, but we show that only one is Stackelberg resilient. \\
    
    Our model is particularly relevant in Web3 scenarios, where these absolute commitments can be realized by the automated and irrevocable nature of smart contracts.  Our work highlights an important issue in ensuring the secure design of Web3. In particular, our work suggests that smart contracts already deployed on major blockchains may be susceptible to these attacks.
\end{abstract}
\clearpage
\pagenumbering{arabic}
\section{Introduction}
We introduce the concept of an \emph{absolute commitment}, where agents have a more `absolute' ability to commit to strategies than is usually the case in games.  In particular, we grant agents the capacity to make irrevocable commitments that can condition on the content of other agents' strategies.  For commitments to be irrevocable, the strategy, once chosen, cannot be altered by the agent. We will generally consider these strategies to be specified using self-executing programs. The ability of commitments to condition on the contents of other commitments makes sense in a context, such as with smart contracts in Web3, in which all commitments are both self-executing and knowable to other agents. 

The \emph{letters of last resort} are secret letters, written by an incoming British prime minister to the commanders of submarines.  They stipulate what the commanders should do in the case that a nuclear strike destroys the British government.  In particular, the letters could call for nuclear retaliation, and even the possibility of such retaliation should act as deterrence.  If we assume compliance by the commanders, we have an example of a self-executing strategy.  It is not subgame perfect to enact meaningless retaliation, so a letter ordering such destruction illustrates a change in equilibrium as compared to a game without irrevocable commitments.  This point was a tenet of mutually assured destruction, a Cold War deterrence strategy that was predicated on the idea that, if one country were annihilated, their submarines would annihilate the aggressor, despite it being too late for such action to save the homeland. 

While this tells part of the story we explore in this paper, there is an aspect missing. To fully encapsulate our narrative, we first require that the letter be known to the potential aggressor and that the aggressor could condition their strategy based on the strategy it contains. In this version, the incoming British prime minister might be told, 
\begin{quote}
    ``\emph{Our spies tell us that our enemy has ordered their submarines that, if you write a letter ordering retaliation in the case of strike, then their submarine commanders are to carry out a strike! We know that enemy spies will know the content of your letter with complete certainty.}''
\end{quote} 
The incoming prime minister now finds themselves in a bind: a letter stipulating retaliation will precipitate an attack and certain retaliation, while a letter forbidding an attack will leave Britain vulnerable. In this story we can already make two observations: there is a decided first mover advantage, Britain here must respond to commitments already made by an enemy; and what was before a protracted game of many moves will now be decided solely in the commitment phase, after which the agents can only sit back and watch the system unfold, knowing precisely what will happen.  

This idea of a leader-follower dynamic is captured by a Stackelberg competition model. In particular, games in which agents can commit to strategies are known as \emph{Stackelberg games} \cite{stackelberg} and their equilibria are known to be hard to compute \cite{stackelberg_complexity,korzhyk2010complexity}. In this model, the leader can do backwards induction to predict what a follower might do in response to their plan. In our scenario, rather than the leader simply being able to deduce what the followers will do in response to their actions, the leader tells the followers precisely what will be done in response to any possible follower strategy.  A regular Stackelberg model returns us to the simple example wherein a leader says that they will retaliate should destruction befall their homeland.  Having made the inference that no follower would then dare attack, the leader can rest assured they will not need their threat.  Conversely, a leader in the case of our absolute commitments would have to follow a different strategy.  If the incoming British prime minister were the leader, rather than responding to an enemy, they may declare, 
\begin{quote}
``\emph{If the enemies of Britain stipulate no retaliation, so will we. Otherwise we will preemptively attack.}'' 
\end{quote}
Such a commitment would clearly lead to the followers complying with the threat and all the submarine crews could go home. These more complex equilibria are known as reverse Stackelberg equilibria \cite{reverse_stackelberg,reverse_stackelberg_2} and are also studied in a variety of contexts, such as in control theory \cite{groot2014toward,groot2017hierarchical,tajeddini2020decentralized}.

In this work, we consider extensive-form games and use subgame perfect equilibria (SPE) as solution concept. We will use the model of absolute commitments proposed by Hall-Andersen and Schwartzbach \cite{smart_contracts}, which strictly generalizes (reverse) Stackelberg equilibria. In this model, granting an agent the capacity to commit absolutely corresponds to allowing this agent to make a `cut' in the game tree, which must respect information sets. This induces an `expanded game' of exponential size, containing a root node belonging to the agent and a subgame corresponding to each possible cut for the agent. For multiple commitments, the commitments are expanded in a bottom-up manner, which gives a natural means for these commitments to condition on each other. Hall-Andersen and Schwartzbach show that reasoning about the subgame equilibria of these games generalizes (reverse) Stackelberg games and is hard in the general case. 

Given the prevalence of known, algorithmically stipulated games, we find this both a practically important and an interesting line of questioning. One clear potential environment for these absolute commitments is in smart contracts \cite{smart_contracts}.  Smart contracts are decentralized programs that run on a virtual machine implemented by a blockchain, such that, once deployed, their execution is no longer under the control of their creator \cite{ethereum}. They are generally used to store and allocate funds, providing clear economic incentives for these attacks.  Importantly, in most blockchains, the contents of smart contracts are public and can, in principle, reason about each other\footnote{Note that while Rice's theorem \cite{rice} states that any non-trivial property of Turing machines is undecidable, in threatening an agent into deploying contract $X$, a smart contract need not check for semantic equivalence, only whether or not the agent deploys exactly contract $X$.}.  This setting was studied by Landis and Schwartzbach in the context of blockchain transaction fee mechanisms \cite{landis2023stackelberg}. In their work, a group of agents have transactions they want included on a block, the contents of which is controlled by a miner. The agents then pay the miner to have their transactions included on the block. However, \cite{landis2023stackelberg} demonstrate that the leading agent may commit absolutely such that their transaction is included at zero cost, and force the other agents to enter into a lottery for the remaining space on the block. In that instance, the outcome benefited all the agents, except for the miner, and the threats were largely just adhering to the regular SPE, whereas the commitments allowed the agents to spontaneously collude to their mutual benefit. By contrast, in this work, we show instances in which the introduction of additional contracts is good for only one or a few of the agents. 

\subsection{Our Results}
As mentioned, games in some contexts may be susceptible to attacks in which agents commit to self-executing strategies that change the nature and equilibrium of the game. We call these attacks \emph{Stackelberg attacks}, since reasoning about these attack captures Stackelberg games as a special case. We observe that some games are resilient to these attacks, in the sense that the set of subgame perfect equilibria is unchanged by adding sequential commitments for the players (in any order). We say that such a game is \emph{Stackelberg resilient}.  We now state the main results of this work.

We first investigate the computational complexity of determining if a game is Stackelberg resilient. We show, using techniques in \cite{smart_contracts}, that Stackelberg resilience is hard to compute in general, but can be computed efficiently in some simple cases.
\begin{theorem}[Computational Complexity, \cref{thm:no_efficient_algo}]
    Determining Stackelberg resilience is \PSPACE-hard in general, although it can be determined efficiently for two-agent games of perfect information.
\end{theorem}
Technically, what we show is that Stackelberg 1-resilience is $\NP$-hard for games of imperfect information, using the same reduction as in \cite{smart_contracts}. However, this result can be extended to show \PSPACE-hardness when the number of agents is unbounded (regardless of whether the games have perfect or imperfect information).

Next, we analyze two escrow mechanisms from the literature \cite{contract_1,contract_2}. These two mechanisms have the same intended functionality: namely, holding a payment in escrow while a trade is being finalized. Interestingly, we find that one of these mechanisms is indeed Stackelberg resilient, while the other one is not. 
\begin{theorem}[Non-Triviality, \cref{thm:contract1_not_resilient,thm:contract2_full_resilient}]
    Stackelberg resilience is non-trivial: there are two escrow mechanisms, only one of which is resilient.
\end{theorem}
Essentially, in one of the games the seller may create a self-executing strategy that forces the buyer to dispute delivery of an item they actually received (causing the buyer to lose their deposit). This demonstrates that Stackelberg resilience is a non-trivial property and begs the question of which mechanisms can be implemented in a Stackelberg resilient manner. We leave it as interesting future work to develop techniques to design mechanisms that are Stackelberg resilient.

Finally, a natural question is whether or not games retain Stackelberg resilience when one agent's ability to make absolute commitments is taken away, i.e. whether or not $k$-resilience implies $(k-1)$-resilience. We call this property \emph{downward closure} and show that it holds in general for Stackelberg resilience.
\begin{theorem}[Downward Closure, \cref{thm:downward_trans}]
    If a game is Stackelberg resilient when $k$ agents can make self-executing strategies, it is also Stackelberg resilient when $\ell$ agents have this capacity for any $\ell \leq k$.
\end{theorem}
We show this by showing the contrapositive statement: a game that is not $(k-1)$-resilient cannot be $k$-resilient. The proof also implies a monotonicity property: once an agent has a viable attack, that attack cannot be undermined by adding an additional self-executing strategy when that commitment is the final one.

\section{Preliminaries}
We assume familiarity with basic game-theoretic concepts, including the definition of extensive-form games and subgame perfection. For the rest of this paper, we will use subgame perfection as the solution and may refer to this simply as an `equilibrium' for brevity. Following the definition of \cite{contract_2}, a pure strategy profile $s^*$ has \emph{$\varepsilon$-strong game-theoretic security} if, for any $i \in [n]$ and any pure $s_i \neq s_i^*$, it holds that,
    $$
        u_i(s_i^*, s_{-i}^*) \geq u_i(s_i, s_{-i}^*) + \varepsilon.
    $$ 
The parameter $\varepsilon$ measures the amount of security. Intuitively, $\varepsilon$ represents the monetary cost to the adversary for deviating from the intended strategy.

\subsection{Absolute Commitments in Games}
We now formalize what we mean by absolute commitments, as introduced by \cite{smart_contracts}. We define an absolute commitment as a commitment that is both irrevocable and that can condition on the content of other agent's commitments.  Absolute commitments can delineate self-executing strategies. Formally, committing to a self-executing strategy instantiates a special type of node with fanout 1 that allows an agent to commit to a strategy in its subgame (see \cref{fig:expansion} for an example). Such nodes are syntactic sugar for a larger tree that contains a node owned by the agent with a subgame for each self-executing strategy they can deploy. For multiple nodes, the self-executing strategies are expanded in a bottom-up manner. In general, computing equilibria in these succinct games is hard, although an efficient algorithm is known for two-agent games of perfect information.

\begin{theorem}[Hall-Andersen, Schwartzbach, \cite{smart_contracts}]\label{thm:pspace}
    Computing the SPE in games with absolute commitments is \PSPACE-hard in general (even restricted to games of perfect information). However, there is a quadratic-time algorithm that computes the SPE in two-commitment games of perfect information. \qed
\end{theorem}

\begin{figure}
    \centering
    \begin{tikzpicture}
   [
    level 1/.style={sibling distance=18mm},
    level 2/.style={sibling distance=14mm},
    level 3/.style={sibling distance=14mm},
    level distance=1.6cm,align=center]
    \node[draw] {$\agentL$}
    child {
        node[draw,circle] {$\agentF$}
        child {
            node[draw, circle] {$\agentL$}
            child { node {$(-\infty,-\infty)$} }
            child { node {$(0,0)$} } }
        child {node {$(1,\,-1)$} }
    };
    \node (A) at (1.85,-1.5) {=};
    \end{tikzpicture}
    \adjustbox{width=0.6\linewidth}{
        \begin{tikzpicture}
        [level 1/.style={sibling distance=25mm},
        level 2/.style={sibling distance=14mm},
        level distance=2cm,align=center,
        every node/.style={thin},
        emph/.style={edge from parent/.style={ very thick,draw}},
        norm/.style={edge from parent/.style={solid,black,thin,draw}}]
        \node[draw,circle] {$\agentL$}
        child[norm] {
            node[draw,circle] {$\agentF$}
            child[emph] {
                node[draw, circle] {$\agentL$}
                child[norm] { node {$(-\infty,-\infty)$} }
                child { node {$(0,0)$} } }
            child[norm] {node {$(1,\,-1)\quad$} }
        }
        child {
            node[draw,circle] {$\agentF$}
            child[emph] {
                node[draw, circle] {$\agentL$}
                child { node {$(0,0)$} } }
            child[norm] {node {$(1,\,-1)\quad$} }
        }
        child[emph] {
            node[draw,circle] {$\agentF$}
            child[norm] {
                node[draw, circle] {$\agentL$}
                child { node {$(-\infty,-\infty)$} }}
            child[emph] {node {$(1,\,-1)$} }
        };
    \end{tikzpicture}
    }
    \caption{
        Expanding an absolute commitment for a simple game.
        The square symbol is a smart contract move for agent $\agentL$.
        We compute all $\agentL$-cuts in the game and connect them with a node belonging to $\agentL$.
        The first coordinate is the leader payoff, and the second is the follower payoff.
        The dominating paths are shown in bold.
        We see that the optimal strategy for $\agentL$ is to commit to choosing $(-\infty, -\infty)$ 
        unless $\agentF$ chooses $(1, -1)$. 
    }
    \label{fig:expansion}
\end{figure}
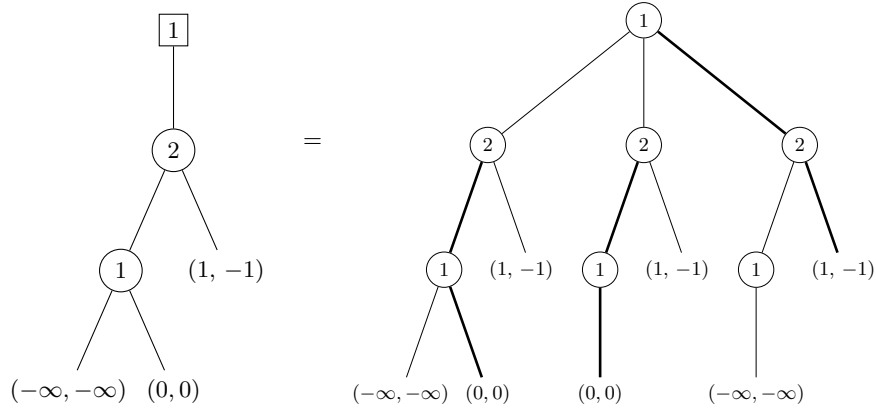

We restate the algorithm, simplified to work only on two-agent games, for self-containment. Let agent 1 have the capacity to commit first, and agent 2 second. Let $A,B \subseteq \mathbb{R}^2$ be two sets of utilities for the two agents and define,
$$
    \textsf{threaten}(A,B) = \{x \in A \mid \exists y \in B. \, y_2 < x_2\},
$$
as the set of nodes from $A$ that agent 1 can threaten agent 2 into accepting, using threats from $B$. The algorithm, which we will call Algorithm \ref{alg}, computes the set of nodes that are inducible by an absolute commitment by the first agent (the inducible region \cite{inducibleregion}) and selects the node that maximizes their utility. For simplicity, assume $G$ is a binary tree. The algorithm recursively computes the nodes that agent 1 can induce in the left and right subtrees and combines them at each branch.

\begin{algorithm}\hspace{-1.125mm} \textsc{\textup{(Hall-Andersen, Schwartzbach, \cite{smart_contracts})}}.\\[2mm]\label{alg}$\textsf{InducibleRegion}(G):$
\begin{enumerate}
    \item[1.] If $G$ is a leaf $\ell$, stop and return $\{\ell\}$.
    \item[2.] If $G$ is a node owned by $i$, with children $G^L$ and $G^R$:
    \begin{enumerate}
        \item[] $I^L \gets \textsf{InducibleRegion}(G^L)$.
        \item[] $I^R \gets \textsf{InducibleRegion}(G^R)$.
        \item[] If $i=1$:
        \begin{enumerate}
            \item[] return $I^L \cup I^R \cup \textsf{threaten}(L, I^L \cup I^R)$.
        \end{enumerate}
        \item[] If $i=2$:
        \begin{enumerate}
            \item[] return $\textsf{threaten}(I^R, I^L)\cup\textsf{threaten}(I^L,I^R)$.
        \end{enumerate}
    \end{enumerate}
\end{enumerate}
\end{algorithm}

Continuing our nuclear example, if Britain has the first commitment move, they could induce favors from other countries as long as they can commit first.  Even in a more complicated game, an agent can make a threat to get something on one side of the tree as long as there exists a reachable threat on the other. 

In practice, running a smart contract requires submitting multiple queries to the blockchain which requires work by the miners and is not free. For the purposes of this work, however, we regard transaction fees as negligible such that we may disregard them entirely. This holds true if the valuations are large compared to the transaction fees. 

\section{Stackelberg Resilient Games}
\label{sec:resilience}
In this section, we formalize what we mean by games that are resilient to Stackelberg attacks. We show that reasoning about such games is hard in general but feasible for two-agent games of perfect information. Let $G$ be an extensive-form game on $n$ agents. We define $\C_i(G)$ as the game that starts with a commitment move for agent $i$ whose only subgame is $G$. Similarly, if $P : [m] \rightarrow [n]$ is a list of agents of length $m$ that specifies the order of the commitments, we denote by $\C_P(G)$ the game with $m$ commitments belonging to the agents specified by the list. Note that $\C_P(G)$ is a succinct representation of a game whose size is exponential in the size of the list. Generally speaking, being the first agent to deploy a commitment is an advantage.  As a result, the equilibrium of the game may be sensitive to the order of the agents in a given list. We call the agent $P(1)$ with the first commitment the \emph{leading commitment agent}. The order of the agents' commitments will depend on the context and is in some cases random. 
Thus, we say a game is Stackelberg resilient if only if the equilibrium remains the same for any order of the commitments. We will now make this notion more formal.

\begin{definition}[Equivalent Games] Let $G, G'$ be two games on $n$ agents. We say that $G$ and $G'$ are \emph{equivalent}, written $G \cong G'$, if for every equilibrium $s^*$ in $G$ (respectively, in $G'$) there exists an equilibrium ${s^*}'$ in $G'$ (respectively, in $G$) such that $u_i(s^*) = u_i({s^*}')$ for every $i\in[n]$. 
\end{definition}

Note that this is an equivalence relation. For two equivalent games, for each equilibrium in either game, there is an equilibrium in the other game with the same payoffs. For auctions, this property is called `weak strategic equivalence'. This definition implicitly suggests that we regard any two outcomes with the same utility vector as equivalent. While this is not necessarily the case in general, it will be the case for the types of games we consider. Namely, in our case we would have $G$ as an extensive-form game in generic form (meaning that all its utility vectors are distinct), and $G'=\C_P(G)$ the same game with contracts in the order specified by $P$. In the game $G'$, we will have multiple copies of each utility vector from $G$, however all its appearances represent the `same' underlying leaf from the game $G$. 

 \begin{figure}
     \centering
     \begin{tikzpicture}
         \node[draw,circle] at (0,0) (g){$2$};
         \node[draw,circle] at (1,-1) (gr){$1$};
         \node at (0.25,-2) (grl){$\bullet$};
         \node at (1.75,-2) (grr){$\bullet$};
         \node at (-1,-1) (gl){$\bullet$};
         \node at (0.125,-2.6) {$\begin{bmatrix*}[l]1:-\infty\\2: -\infty\end{bmatrix*}$};
\node at (1.75,-2.6) {$\begin{bmatrix*}[l]1:0\\2: 0\end{bmatrix*}$};
         \node at (-1,-1.6) {$\begin{bmatrix*}[l]1:1\\2: -1\end{bmatrix*}$};
         \draw[line width=0.5mm] (g) to (gr);
         \draw (g) to (gl);
         \draw (gr) to (grl);
         \draw[line width=0.5mm] (gr) to (grr);
     \end{tikzpicture}
     \caption{An example of a game that is not 1-resilient (to see explicitly why, see \cref{fig:expansion}). Agent $2$ has a coin that they can choose to give to agent $1$. Agent $1$ is subsequently given the option to trigger nuclear annihilation. Without absolute commitments, the SPE is the node $(0,1)$ where agent $2$ keeps the coin because nuclear annihilation is an empty threat. However, when agent $1$ has an absolute commitment, they will automatically retaliate if they do not receive the coin, thus changing the SPE to $(1,0)$. Such an equilibrium is called a Stackelberg equilibrium.  Note that in this example the inducible region for $i$ is everything but nuclear annihilation.}
     \label{fig:example}
 \end{figure}
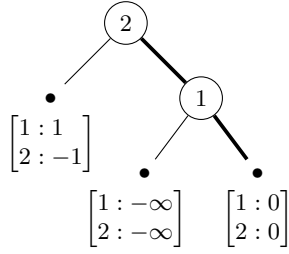

\begin{definition}[Stackelberg Resilience]
    A game $G$ is said to be \emph{Stackelberg $k$-resilient} for an integer $k>0$ if, for any list $P$ of $k$ distinct agents, it holds that $\C_P(G) \cong G$. We say $G$ is \emph{(full) Stackelberg resilient} if it is Stackelberg $n$-resilient.
\end{definition} 

Note that if the SPE of a game is unique and labeled with some utility vector $u \in \mathbb{R}^n$, then Stackelberg resilience says that every SPE in the expanded tree also has to be labeled with $u$. We emphasize that we are \emph{not} assuming uniqueness of the equilibria -- the above definition simply requires the set of equilibria to be unchanged. However, sometimes, it will be convenient for the purpose of analysis to assume the equilibrium is unique -- this is only for simplicity of presentation, and we stress that the results of this work are not contingent on the equilibria being unique. Similarly, we do not require the games to have perfect information, the definition extends easily to games of imperfect information.

In this definition, we require the list $P$ to consist of distinct agents (i.e. $P$ is injective). If this were not the case, $(k+1)$-resilience would trivially imply $k$-resilience for the uninteresting reason that adding absolute commitments for a fixed agent is an idempotent operation, i.e. the same agent having two nested commitments is equivalent to them having only the topmost commitments. Also note that if a game is not Stackelberg resilient, there exists a Stackelberg attack that some agent can commit to in order to obtain better utility (or force worse utility for others) as compared to the situation without commitments. Only Stackelberg resilient games are not susceptible to Stackelberg attacks.

There are a few observations that we can see immediately.  First, if every agent has the same most preferred outcome and this outcome is the SPE of the original game, there cannot be a viable attack and the game is trivially resilient.  We also observe that if an agent has the last commitment and their only node is the root of the original game tree, then the choice of absolute commitment and the choice of move `collapse' and they cannot affect change through their commitment move. 

In terms of computational complexity, we note that the reduction used for Theorem 1 in \cite{smart_contracts} works also for Stackelberg resilience. Moreover, Stackelberg resilience can be computed for two-commitment games of perfect information by invoking Algorithm \ref{alg} twice, relabeling the contract players in the second invocation. This establishes the following result.

\begin{theorem}\label{thm:no_efficient_algo}
    It is $\NP$-hard to compute Stackelberg 1-resilience for games with more than two agents. However, full Stackelberg resilience can be computed efficiently for two-commitment games of perfect information.
\end{theorem}
\begin{proof}[Sketch]
Follows using the same reduction as in the proof of \cref{thm:pspace} from \cite{smart_contracts}.  The reduction is from \textsc{CircuitSAT} and constructs a 3-agent game of imperfect information, in such a way that agent 1 obtains good utility by making absolute commitments if and only if the circuit is satisfiable. In particular, agent 1 has a small gadget for every input to the circuit where they choose the value of the input using their strategy. agent 1 prefers the value of true to the value of false. agents 2 and 3 simulate the logic of the circuit. The resulting game has the property that, when no agent has access to absolute commitments, by backward induction, agent 1 will assign $\top$ to every gadget to obtain zero utility in the SPE, while with one agent with absolute commitments there is a self-executing strategy for agent 1 in which they get 1 utility in the SPE if and only the input circuit is satisfiable. In addition, since agent 1 always moves last and thus has to play the SPE, the equilibrium of the game does not change by giving either agent 2 or agent 3 a absolute commitments, as they obtain their preferred outcome among the set of feasible outcomes given that agent 1 always chooses $\top$. Thus, the game is not 1-resilient if and only if the circuit is satisfiable, which shows that computing 1-resilience is \NP-hard for three-agent games. For the latter part, note that by \cref{thm:pspace}, full Stackelberg resilience can be computed efficiently for games of two agents by invoking \textsf{InducibleRegion} twice, relabeling the agents in the second invocation. \qed
\end{proof}

\section{Non-Triviality of Resilience}
\label{sec:commerce}
In this section, we demonstrate non-triviality of Stackelberg resilience by analyzing two escrow mechanisms from the literature \cite{contract_1,contract_2}. Both of these mechanisms enable a buyer and a seller to exchange a good, and are intended to be deployed as smart contracts on a blockchain. In such a setting, the agents natively have the capacity to deploy their own smart contracts, amounting to making an absolute commitment and implementing it as a self-executing strategy.  Both mechanisms involve a seller $S$ and a buyer $B$ that want to exchange an item \emph{it} for a price of $x$. In accordance with \cite{contract_2}, we let $y$ denote the value of \emph{it} to the buyer and $x'$ the value to the seller, and assume $y>x>x'>0$. Both mechanisms are shown to securely implement an escrow functionality with $\varepsilon$-strong game-theoretic security for arbitrarily large $\varepsilon$ (and suitable parameters). The mechanisms are depicted in \cref{fig:gametree2,fig:gametree3}.

The first contract by Asgaonkar and Krishnamachari \cite{contract_1} assumes \emph{it} is digital and can be used as input in a cryptographic hash function. We present a slightly simplified version\footnote{In \cite{contract_1}, $S$ and $B$ also have the option of submitting garbage to the contract. The option is inconsequential to the analysis and has been removed for brevity.}.

\begin{figure}
	\centering
        \begin{scaletikzpicturetowidth}{\columnwidth}
		\begin{tikzpicture}[scale=\tikzscale]
		\begin{scope}[every node/.style={circle,thick,draw}]
		\node (A) at (3.75,0) {$S$}; 
		\node (B) at (1.2,-1.5) {$B$}; 
		\node (H) at (6.3,-1.5) {$B$}; 
		\end{scope}
		\node (D) at (2.6,-4) {$\bullet$};
		\node at (2.6,-4.5) {$\begin{bmatrix*}[l]\buyer y-x\\\seller  x\end{bmatrix*}$};
		\node (F) at (7.5,-4) {$\bullet$};
		\node at (7.5,-4.5) {$\begin{bmatrix*}[l]\buyer \hfill0\\\seller -\lambda\end{bmatrix*}$};
		\node (I) at (4.7,-4) {$\bullet$};
		\node at (4.7,-4.5) {$\begin{bmatrix*}[l]\buyer -x-\lambda\\\seller x\end{bmatrix*}$};
		\node (K) at (0,-4) {$\bullet$};
            \node at (0,-4.5){$\begin{bmatrix*}[l]\buyer y-x-\lambda\\\seller x\end{bmatrix*}$};
		
		\begin{scope}[>={Stealth[black]},
		every node/.style={fill=white,rectangle}]
		\path [-] (A) edge[line width=0.5mm] node {send} (B);
		\path [-] (A) edge node {not send}  (H);
		\path [-] (B) edge[] node {dispute}  (K);
		\path [-] (B) edge[line width=0.5mm] node {accept} (D);
		\path [-] (H) edge[line width=0.5mm] node {dispute} (F);
		\path [-] (H) edge node {accept} (I);
		\end{scope}
		\end{tikzpicture}
        \end{scaletikzpicturetowidth}
  
	\caption{Mechanism \ref{contract_1} represented as an extensive-form game. Each branch is labeled with the agent who owns the node. We denote by $u^j$ the $j^\text{th}$ utility vector from left-to-right, for $j=1\ldots 4$.}
	\label{fig:gametree2}
\end{figure}
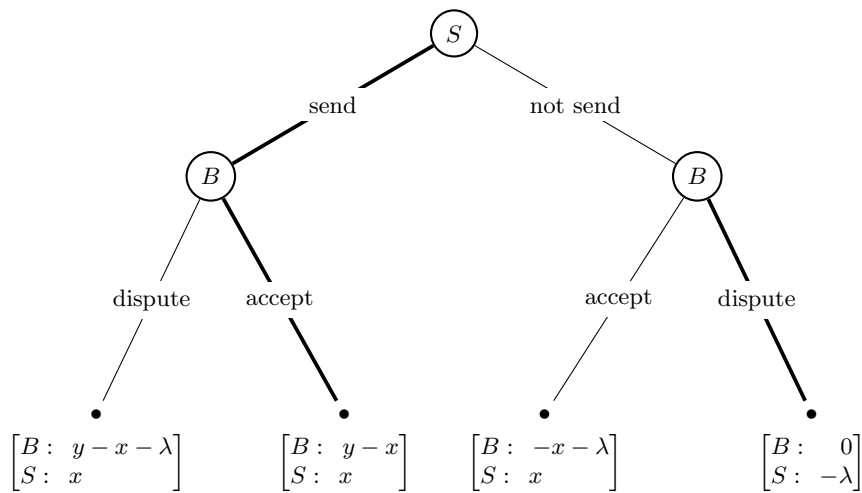

\begin{contract}\hfill\label{contract_1}
\begin{enumerate}
    \item[1.] $S$ makes public $h \gets H(it)$ where $H$ is a hash function, and deposits $\lambda$ money to the contract.
    \item[2.] $B$ deposits $x+\lambda$ money to the contract.
    \item[3.] $S$ either submits $it$ (or an encryption thereof) to the contract, or submits a wrong item.
    \item[4.] $B$ either \textbf{accepts} or \textbf{disputes} delivery of the item.
    \begin{enumerate}
        \item[4.1.] If $B$ accepted, $\lambda$ money is given to $B$ and $x+\lambda$ money given to $S$, and the contract terminates.
        \item[4.2.] If $B$ disputed, the contract recomputes $H(it)$ and compares with $h$.
        \begin{enumerate}
            \item[4.2.1.] If equal, it forwards $x+\lambda$ money to $S$.
            \item[4.2.2.] If unequal, it forwards $x+\lambda$ money to $B$.
        \end{enumerate}
    \end{enumerate}
\end{enumerate}
\end{contract}
\begin{theorem}[Asgaonkar, Krishnamachari \cite{contract_1}]\label{thm:non_trivial_1}
    For any $\varepsilon\geq 0$, and any sufficiently large $\lambda>0$, Mechanism \ref{contract_1} has $\varepsilon$-strong game-theoretic security. 
\end{theorem}

Another contract by Schwartzbach \cite{contract_2} does not assume \emph{it} can be hashed or processed by the mechanism. Instead, it makes use of an oracle to assess which agent is the honest one. As the oracle is assumed expensive, it is only invoked optimistically in case of disputes. Let $\gamma$ denote the error rate of the oracle.

\begin{contract}\hfill\label{contract_2}
\begin{enumerate}
    \item[1.] $B$ deposits $x$ money to the smart contract.
    \item[2.] $S$ sends \emph{it} to $B$ (off-chain).
    \item[3.] $B$ either \textbf{accepts} or \textbf{disputes} delivery, in which case they deposit $\lambda$ money.
    \begin{enumerate}
        \item[3.1.] If $B$ accepted, $x$ money is forwarded to $S$ and the contract terminates.
        \item[3.2.] If $B$ disputed, $S$ can either \textbf{forfeit} or \textbf{counter} the dispute, in which case they also deposit $\lambda$ money.
        \begin{enumerate}
            \item[3.2.1.] If $S$ forfeited, $x+\lambda$ money is returned to $B$.
            \item[3.2.2.] If $S$ disputed, the oracle is invoked. Whomever is deemed honest by the oracle receives back $x+\lambda$ money.
        \end{enumerate}
    \end{enumerate}
\end{enumerate}
\end{contract}

\begin{theorem}[Schwartzbach, \cite{contract_2}]\label{thm:non_trivial_2}
    If $\gamma<\frac12$ and $\lambda=x$, then Mechanism \ref{contract_2} has $x\,(1-2\gamma)$-strong game-theoretic security.
\end{theorem}

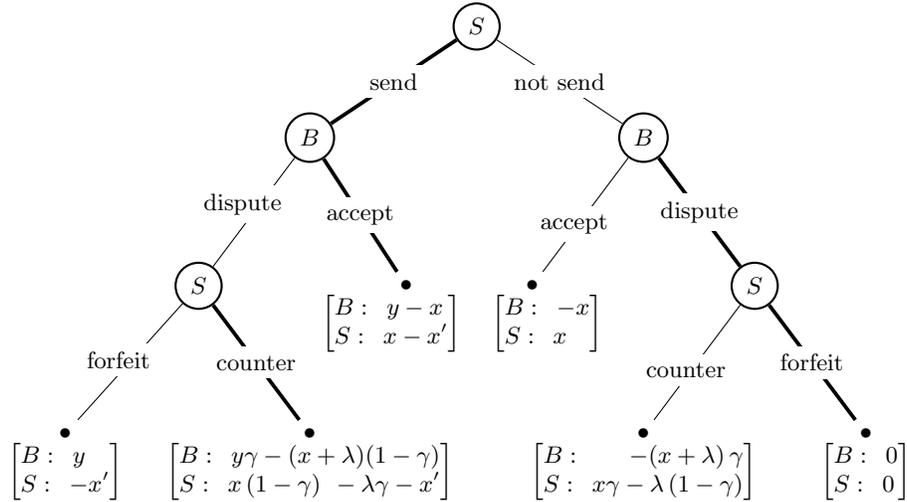
\begin{figure}
	\centering
        \begin{scaletikzpicturetowidth}{\columnwidth}
		\begin{tikzpicture}[scale=\tikzscale]
		\begin{scope}[every node/.style={circle,thick,draw}]
		\node (A) at (3.75,0.5) {$S$}; 
		\node (B) at (1.5,-1) {$B$}; 
		\node (E) at (0,-3) {$S$}; 
		\node (H) at (6,-1) {$B$}; 
		\node (J) at (7.5,-3) {$S$}; 
		\end{scope}
		\node (D) at (2.8,-3) {$\bullet$};
		\node  at (2.6,-3.5) {$\begin{bmatrix*}[l]\buyer y-x\\\seller x-x'\end{bmatrix*}$};
		\node (F) at (-1.8,-5) {$\bullet$};
		\node at (-1.8,-5.5) {$\begin{bmatrix*}[l]\buyer y\\\seller -x'\end{bmatrix*}$};
		\node (G) at (1.5,-5) {$\bullet$};
		\node at (1.5,-5.5) {$\begin{bmatrix*}[l] \buyer y \gamma-(x+\lambda)(1-\gamma)\\\seller x\,(1-\gamma) \hfill - \lambda\gamma-x'\end{bmatrix*}$};
		\node (I) at (4.5,-3) {$\bullet$};
		\node  at (4.7,-3.5) {$\begin{bmatrix*}[l]\buyer -x\\\seller x\end{bmatrix*}$};
		\node (L) at (6,-5) {$\bullet$};
		\node at (6,-5.5) {$\begin{bmatrix*}[l]\buyer\hfill-(x+\lambda)\,\gamma\\\seller x\gamma-\lambda\,(1-\gamma)\end{bmatrix*}$};
		\node (K) at (9,-5) {$\bullet$};
		\node at (9,-5.5) {$\begin{bmatrix*}[l]\buyer 0\\\seller 0\end{bmatrix*}$};
		
		\begin{scope}[>={Stealth[black]},
		every node/.style={fill=white,rectangle}]
		\path [-] (A) edge[line width=0.5mm] node {send} (B);
		\path [-] (A) edge node {not send}  (H);
		\path [-] (B) edge node[xshift=-0.15cm,yshift=0.1cm] {dispute}  (E);
		\path [-] (B) edge[line width=0.5mm] node {accept} (D);
		\path [-] (E) edge node[xshift=-0.15cm,yshift=0.025cm] {forfeit} (F);
		\path [-] (E) edge[line width=0.5mm] node {counter} (G);
		\path [-] (H) edge node[xshift=-0.15cm,yshift=-0.125cm] {accept} (I);
		\path [-] (H) edge[line width=0.5mm] node {dispute} (J);
		\path [-] (J) edge[line width=0.5mm] node {forfeit} (K);
		\path [-] (J) edge node[xshift=-0.15cm,yshift=-0.1cm]{counter} (L);
		\end{scope}
		\end{tikzpicture}
        \end{scaletikzpicturetowidth}
	\caption{Mechanism \ref{contract_2} represented as an extensive-form game. Each branch is labeled with the agent who owns the node. We denote by $u^j$ the $j^\text{th}$ utility vector from left-to-right, for $j=1\ldots 6$.}
	\label{fig:gametree3}
\end{figure}

Although the mechanisms are similar, only one of them is Stackelberg resilient, which demonstrates that Stackelberg resilience is a non-trivial property.
\begin{theorem}\label{thm:contract1_not_resilient}
  Mechanism \ref{contract_1} is not Stackelberg resilient.
\end{theorem}
\begin{proof}
Label the utility vectors at the leaves of contract, seen as the tree in \cref{fig:gametree2}, $u^j$ for $j \in [4]$ from left to right, where $u^1$ is the leftmost leaf, corresponding to strategy send-dispute, and so on across the base of the tree.  Let $u^j_i$ for $i\in \{B,S\}$ be the utility of the buyer and seller at leaf $j$. Consider the case where the first agent is $S$. We apply Algorithm \ref{alg} to the tree and keep track of set $I$ for each node as we move up the tree. At the leaf level, each leaf is its own set $I$. The next level up depends on the left and right child, $G^L$ and $G^R$, respectively, of the root. We denote $I^L$ and $I^R$ to be the associated sets. Both nodes are owned by the buyer, which is agent 2. It is easy to see that $u_B^2>u_B^1$, which yields $I^L = \{u^2\} $, and $u_B^4>u_B^3$, which yields $I^R =\{u^4\}$, as the inducible regions. Moving up the tree, we arrive at the root, which is owned by the seller. Here the seller can use any of the $I^L \cup I^R$ utilities to threaten the buyer into any of the $L$ leaves. Thus the seller can use $u^4 \in I^R$ to threaten the buyer into $u^1$ as long as $y-x-\lambda > 0$, which is allowed for by the stipulations on parameters in \cite{contract_1}. So $I = \{u^1,u^2,u^4\}$. The seller has equal utility in $u^1$ and $u^2$ and, since we assume the agents are weakly malicious and $u_B^2>u^1_B$, the seller will choose $u^1$, which is not the SPE. Thus the contract is not Stackelberg 2-resilient.  
    If instead $B$ has the first contract, the SPE and reverse Stackelberg equilibria coincide. Since the buyer owns the middle level of nodes, we have $I^L = \{u^1, u^2\}$ and $I^R = \{u^3, u^4\}$.  At the root, the buyer can threaten the seller into $u^1$ and $u^2$ with $u^4$ and into $u^3$ with either element of $S^L$ because the seller is weakly malicious. From $S=\{u^1, u^2, u^3\}$, the buyer will pick $u^2$, which is the SPE.\qed
\end{proof}
Note the $S$ Stackelberg attack would still be viable if we accounted for the loss of the sale item in $u^1_S$ and $u^2_S$, as is the case in the game from \cite{contract_2}. In the full version of the contract in \cite{contract_1}, which allows agents to additionally play a garbage string, it can be readily observed that $S$ can threaten with garbage to get the $u^1$ equivalent regardless of the value of the deposit $\lambda$.
 
\begin{theorem}\label{thm:contract2_full_resilient}
   Mechanism \ref{contract_2} is full Stackelberg resilient.
\end{theorem}
\begin{proof}
In keeping with \cref{thm:contract1_not_resilient}, we label the leaves $u^j$, $j \in [6]$ from left to right, with reference to \cref{fig:gametree3}. Let $I^{LL}$ be the set associated with left child of the left child of the root and $I^{RR}$ symmetrically so on the right.

Suppose first that $S$ is agent 1. Since both nodes at the third level, $G^{LL}$ and $G^{RR}$ belong to the seller, we have $I^{LL} = \{ u_1, u_2\}$ and $I^{RR} = \{ u_5, u_6\}$.  The next level is comprised of nodes owned by the buyer.  On the LHS, the seller can use $u^2$ to threaten the buyer into $u^3$ and $u^3$ to threaten for $u^1$ because $u^1_B>u^3_B>u^2_B$. We have $u^3_B>u^2_B$ if $y-x > y \gamma - (x+\lambda)(1-\gamma)$. Since $\gamma < 1/2$, the desired result is true if $y-x> y - (x+\lambda)(1-\gamma)$.  Simplifying and rearranging yields $\frac{\gamma}{1-\gamma}x < \lambda $, which is a stipulation of \cite{contract_2}. So we have $I^L=\{u^1,u^3\}$.  On the RHS, we can readily see $u^6_B>u^4_B$. We also have $u^5_B>u^4_B$ if $-x < -(x+\lambda)\gamma$. Solving this expression for $\lambda$, we have $\frac{1-\gamma}{\gamma}x>\lambda$, which is again required by the contract. So $B$ will always want to move left, yielding $I^R = \{ u^5, u^6 \}$. At the root, $S$ can threaten using any strategy in $I^L \cup I^R$. It is easy to see that the element with the lowest utility for $B$ is $u^5$.  We have already shown that $u^5_B>u^4_B$. Thus $I = \{ u^1, u^3, u^5, u^6\} \cup \{ u^2 \mid u^5_B<u^2_B, \}$, that is, $u^2$ is in $I$ if $u^5_B<u^2_B$. While it is true that $u^5_B>u^2_B$, it is easier to see that even if $u^2$ were inducible, it would be a less optimal result for $S$ compared to $u^3$, a fact easily seen with the observation that $\gamma <1/2$. Of the remaining choices, it immediate that $u^1$ and $u^6$ will not be optimal strategies for $S$. All that is left to show is $u^3_S > u^5_S$.  The desired result is true if $x-x'> x\gamma -\lambda(1-\gamma)$, which can be rearranged to $x> \frac{1}{1-\gamma}x'-\lambda$. Since we have the assumption that $x>x'$, the previous statement is true if $ x > \frac{1}{1-\gamma}x-\lambda$, given that the coefficient $\frac{1}{1-\gamma}$ must be positive.  Solving for $\lambda$ yields $\lambda > \frac{\gamma}{1-\gamma}x $, which is a requirement of \cite{contract_2}. Thus, $S$ will pick $u^5$, the SPE in \cite{contract_2}.

Now suppose that $B$ has the first contract. $S$ owns $G^{LL}$ and $G^{RR}$ and thus $I^{LL}=\{ u^2\}$ because $u^2_S > u^1_S$.  We have $I^{RR}= \{ u^6\} $ because $u^6_S > u^5_S$, which can readily be seen to be true given the condition  $\lambda > \frac{\gamma}{1-\gamma}x $.  It is easy to see that $I^L = \{ u^2, u^3 \}$ because $u^3_S$ is the largest of the three utilities for $S$ and cannot be used to threaten for any other. On the RHS, we have an analogous situation; it is immediate that $u^4$ cannot be used for threats. Thus $I^R = \{ u^4, u^6 \}$. At the root, we notice that neither $u^4$ nor $u^6$ can be used to threaten for $u^1$. Thus, ${u^2, u^3, u^4, u^6} \subseteq I$ and $u^1 \not\in I$.  It is not clear if $u^5_S>u^2_S$, but we can immediately see that $u^5$ will not be the optimal choice for $B$ if it is in $I$. In fact we see that the only possibly positive inducible choices for $B$ are $u^2$ and $u^3$. We proved $u^3_B > u^2_B$ above. Thus B picks $u^3$, which again coincides with the SPE. Since both arrangements of contracts yield the SPE, it follows from \cref{thm:1resilience} that the game is full Stackelberg resilient.\qed
\end{proof}

\begin{theorem}\label{thm:1resilience}
   Both contracts are Stackelberg 1-resilient.
\end{theorem}
\begin{proof} 
For Mechanism \ref{contract_2}, the result follows from \cref{thm:downward_trans,thm:contract2_full_resilient}. For Mechanism \ref{contract_1}, if the sole contract is owned by $S$, there is no  buyer move between when $S$ determines its contract and when $S$ moves and has no impact. Thus the game reduces to the SPE. By the proof of \cref{thm:downward_trans}, given that the order $B$-$S$ is 2-resilient, we have that a $B$-contract will coincide with the SPE.\qed
\end{proof}

In order for a Stackelberg attack to be feasible and worthwhile, there needs to be a reachable threat and a more desirable outcome for the threatening agent.  In Mechanism \ref{contract_1}, we see that threatening not to send is a viable threat against the buyer, regardless of what the buyer later plays. This can be used to threaten the buyer into an erroneous dispute resulting in them losing their deposit. Mechanism \ref{contract_2} has the further mechanism of the oracle, which both weakens the threat and removes the incentive for $S$ to attempt to instigate a different outcome. Since the oracle has some error rate, there is a chance that it punishes $S$ if $B$ untruthfully disputes, thus removing any benefit to $S$ of the threat.  The threat is also no longer viable, that is $u_B^5>U_B^2$, given further conditions in \cite{contract_2}.  Thus Mechanism \ref{contract_2} frustrates this type of attack both from the demand and the supply.

\section{Downward Closure}
It is a natural question to wonder if, given resilience in the case with $k$ commitment agents, i.e. agents with the capacity to make absolute commitments, we have resilience in $(k-1)$ case, and thus, inductively, for all subsequent removals of self-executing strategies. The contrapositive of this question is also interesting in its own right: can resilience be restored by adding a self-executing strategy? What if this added self-executing strategy is in the least favorable, last position? Indeed, we need to address the contrapositive to answer the original question.

Consider a game $G^0$ that is not $(k-1)$-resilient.  Then some agent has a Stackelberg attack against the others that results in a better equilibrium for the attacker. These attacks work by allowing the attacker to commit to a worse outcome if the others do not comply. Thus, the attacking agent can coerce the other agents into obeying some threat.  We must now ask if there can be some sort of \emph{threat to the threat}.  As we will illustrate in the following example, there can indeed be a threat to the threat if the new commitment agent gets to go first.  On the other hand, if the new commitment agent has the last commitment, they cannot threaten the original threat.  To intuitively see why going last nullifies any potential threat to the threat, recall that the last commitment is effectively the last move.  Thus, if this final $k^{\text{th}}$ agent could make a threat, it cannot commit until it is too late, in some sense `keeping the threat a secret' until all other agents have already made their moves.  This means that if the $k^{\text{th}}$ agent makes a threat, they know if they will have to play it. 

\begin{figure}
    \centering
    \begin{tikzpicture}
        \node[draw,circle] at (0,0) (g){$i$};
        \node[draw,circle] at (1.5,-1) (gr){$j$};
        \node at (2.2,-2.85) {$\begin{bmatrix*}[l]i: \hfill10\\j: \hfill -1\\k: -10\end{bmatrix*}$};
        \node at (0.625,-2.85) {$\begin{bmatrix*}[l]i: \hfill0\\j: \hfill 10\\k: 10\end{bmatrix*}$};
        \node at (-0.875,-2.85) {$\begin{bmatrix*}[l]i: \hfill-10\\j: \hfill -10\\k: \hfill 0\end{bmatrix*}$};
        \node at (-2.375,-2.85) {$\begin{bmatrix*}[l]i: \hfill-1\\j: \hfill 0\\k: -1\end{bmatrix*}$};
        \node at (0.75,-2) (grl){$\bullet$};
        \node at (2.25,-2) (grr){$\bullet$};
        \node at (-0.75,-2) (glr){$\bullet$};
        \node at (-2.25,-2) (gll){$\bullet$};
        \node[draw,circle] at (-1.5,-1) (gl) {$k$};
        \draw[line width=0.5mm] (g) to (gr);
        \draw (g) to (gl);
        \draw[line width=0.5mm] (gr) to (grl);
        \draw[] (gr) to (grr);
        \draw (gl) to (gll);
        \draw[line width=0.5mm] (gl) to (glr);
    \end{tikzpicture}
    \caption{An example of a game that is 1-resilient, but neither 2-resilient nor 3-resilient.  We denote by $u^\ell$ the $\ell^\text{th}$ utility vector from left-to-right, for $\ell=1\ldots 4$.}
    \label{fig:example_G0}
\end{figure}
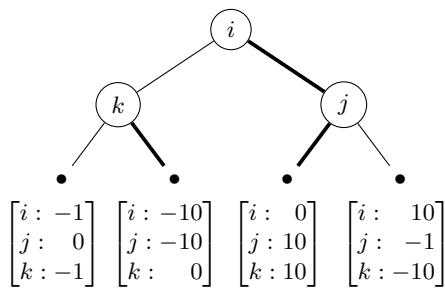

To illustrate this, we introduce the example illustrated in \cref{fig:example_G0}. We label the utility vectors $u^1, \dotsc , u^4$ from left to right for ease of notation. We note the following useful inequalities.
\begin{align*}
    u^4_i &> u^3_i > u_i^1 > u_i^2,\\
    u^3_j &> u^1_j > u_j^4 > u_j^2,\\
    u^3_k &> u^2_k > u^1_k > u_k^4.
\end{align*}
The SPE in this game is easy to read off: we can see that if $k$ gets a choice, they will chose to play right, resulting in $u^2$ and, if $j$ gets a choice, they will chose left, resulting in $u^3$. Seeing this, agent $i$ will choose right, yielding $u^3$ as the SPE. The addition of just one commitment cannot change the equilibrium.  Neither agent $j$ nor $k$ could do better than the SPE and, since $i$ owns the root, committing to a move and going first are effectively the same.  A second commitment can break resilience.  If $i$ has the first commitment and $j$ the second, we can have the following threat from $i$:
\begin{quote}
    \begin{enumerate}
        \item[$i$:] ``\emph{If agent $j$ does not \emph{commit to} playing right, I will play left.}'' 
    \end{enumerate}
\end{quote}
This commitment can be converted into a cut for agent $i$: cut away the right branch in every subgame where agent $j$ did not cut away their left branch. Moving forward, we will not formally convert the commitments into cuts and trust that the intention is clear.

If $j$ does not comply, then $i$ plays left and $k$ will play right, resulting in $u^2$.  If $j$ does comply, the outcome will be $u^4$, which is a better option for $j$ than $u^2$.  Thus, $j$ must comply. We label this game, with the specific commitment order of $i$ then $j$, $G^1$. So far, agent $k$ has had no impact on $i$'s antics as both other agents know that, should the game come to $k$'s node, $k$ has no choice but to play right for $u^2$. In fact, $i$'s threat is predicated on this.

Both $j$ and $k$ are worse off in $G^1$ than in $G^0$.  If $k$ gets the first commitment, before those of $i$ and $j$, the $G^1$ threat can be nullified.  agent $k$ commits to a \emph{threat to the threat} wherein their commitment commits them to play left if $i$ plays the commitment from $G^1$. This means that $i$'s `threat' results now $u^1$, not $u^2$, and $j$ would prefer $u^1$ to $u^4$. Thus if $i$ deployed the commitment from $G^1$, $j$ would not be threatened into committing the commitment stipulated by $i$. Now $i$ can infer that if they try the $G^1$ commitment as second agent after $k$, the equilibrium will be $u^1$, which is worse for $i$ than the old SPE $u^3$.  With no viable threat, we end up on $u^3$ again.

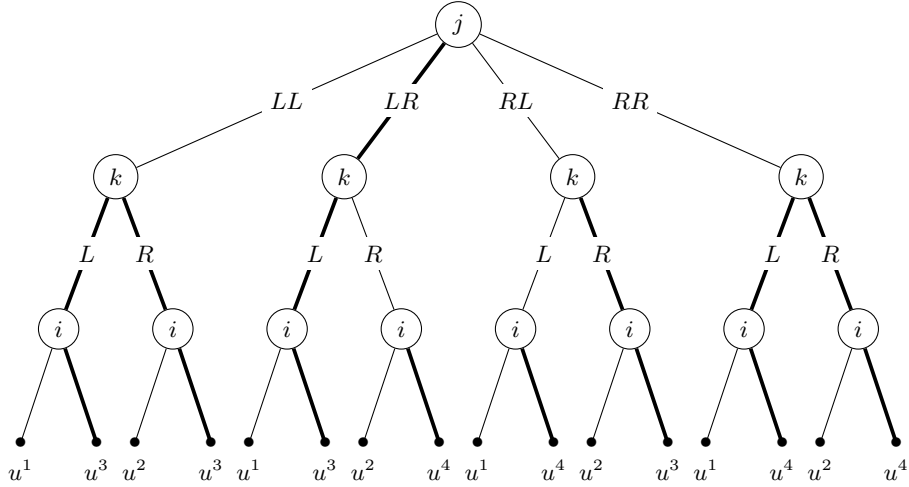
\begin{figure}
    \centering
    \begin{scaletikzpicturetowidth}{\columnwidth}
    \begin{tikzpicture}[scale=\tikzscale]
        \node[draw,circle] at (0,1) (g) {$j$};
        \node[draw,circle] at (-4.5,-1) (gLL) {$k$};
        \node[draw,circle] at (-1.5,-1) (gLR) {$k$};
        \node[draw,circle] at (1.5,-1) (gRL) {$k$};
        \node[draw,circle] at (4.5,-1) (gRR) {$k$};
        \draw (g) --node[fill,white,text=black]{$LL$} (gLL);
        \draw[line width=0.5mm] (g) -- node[fill,white,text=black]{$LR$} (gLR);
        \draw (g) --node[fill,white,text=black]{$RL$} (gRL);
        \draw (g) --node[fill,white,text=black]{$RR$} (gRR);
        \node[draw,circle] at (-5.25,-3) (gLLL) {$i$};
        \draw[line width=0.5mm] (gLL) -- node[fill,white,text=black]{$L$} (gLLL);
        \draw (gLLL) -- (-5.75,-4.5) node {$\bullet$};
        \node at (-5.75,-4.85) {$u^1$};
        \draw[line width=0.5mm] (gLLL) -- (-4.75,-4.5) node {$\bullet$};
        \node at (-4.75,-4.85) {$u^3$};
        \node[draw,circle] at (-3.75,-3) (gLLR) {$i$};
        \draw[line width=0.5mm] (gLL) -- node[fill,white,text=black]{$R$} (gLLR);
        \draw (gLLR) -- (-4.25,-4.5) node {$\bullet$};
        \node at (-4.25,-4.85) {$u^2$};
        \draw[line width=0.5mm] (gLLR) -- (-3.25,-4.5) node {$\bullet$};
        \node at (-3.25,-4.85) {$u^3$};
        \node[draw,circle] at (-2.25,-3) (gLRL) {$i$};
        \draw[line width=0.5mm] (gLR) --node[fill,white,text=black]{$L$} (gLRL);
        \draw (gLRL) -- (-2.75,-4.5) node {$\bullet$};
        \node at (-2.75,-4.85) {$u^1$};
        \draw[line width=0.5mm] (gLRL) -- (-1.75,-4.5) node {$\bullet$};
        \node at (-1.75,-4.85) {$u^3$};
        \node[draw,circle] at (-0.75,-3) (gLRR) {$i$};
        \draw (gLR) --node[fill,white,text=black]{$R$} (gLRR);
        \draw (gLRR) -- (-1.25,-4.5) node {$\bullet$};
        \node at (-1.25,-4.85) {$u^2$};
        \draw[line width=0.5mm] (gLRR) -- (-0.25,-4.5) node {$\bullet$};
        \node at (-0.25,-4.85) {$u^4$};
        \node[draw,circle] at (0.75,-3) (gRLL) {$i$};
        \draw (gRL) --node[fill,white,text=black]{$L$} (gRLL);
        \draw (gRLL) -- (0.25,-4.5) node {$\bullet$};
        \node at (0.25,-4.85) {$u^1$};
        \draw[line width=0.5mm] (gRLL) -- (1.25,-4.5) node {$\bullet$};
        \node at (1.25,-4.85) {$u^4$};
        \node[draw,circle] at (2.25,-3) (gRLR) {$i$};
        \draw[line width=0.5mm] (gRL) --node[fill,white,text=black]{$R$} (gRLR);
        \draw (gRLR) -- (1.75,-4.5) node {$\bullet$};
        \node at (1.75,-4.85) {$u^2$};
        \draw[line width=0.5mm] (gRLR) -- (2.75,-4.5) node {$\bullet$};
        \node at (2.75,-4.85) {$u^3$};
        \node[draw,circle] at (3.75,-3) (gRRL) {$i$};
        \draw[line width=0.5mm] (gRR) --node[fill,white,text=black]{$L$} (gRRL);
        \draw (gRRL) -- (3.25,-4.5) node {$\bullet$};
        \node at (3.25,-4.85) {$u^1$};
        \draw[line width=0.5mm] (gRRL) -- (4.25,-4.5) node {$\bullet$};
        \node at (4.25,-4.85) {$u^4$};
        \node[draw,circle] at (5.25,-3) (gRRR) {$i$}; \draw[line width=0.5mm] (gRR) --node[fill,white,text=black]{$R$} (gRRR);
        \draw (gRRR) -- (4.75,-4.5) node {$\bullet$};
        \node at (4.75,-4.85) {$u^2$};
        \draw[line width=0.5mm] (gRRR) -- (5.75,-4.5) node {$\bullet$};
        \node at (5.75,-4.85) {$u^4$};
    \end{tikzpicture}
    \end{scaletikzpicturetowidth}
    \caption{The game \cref{fig:example_G0} expanded with a commitment for $j$ and then for $k$. All subgames consisting only of moves for $j$ and $k$ have been collapsed to the SPE. Agent $j$ can commit to going $L$ or $R$ depending on whether or not $k$ commits to $L$ or $R$. Now suppose we add a commitment for agent $i$ at the beginning. By making appropriate cuts in this tree, agent $i$ can commit to actions that force $j$ to commit to $RR$. }
    \label{fig:example_G1}
\end{figure}

However, suppose we allow $k$ only the last commitment, yielding the order $i,j,k$. This order means that $k$ is still last to move. Intuitively, the reason that the \emph{threat to the threat} will not work is that $k$ moves last and cannot commit to an action that they know will leave them worse off.  To make this precise, we first expand the game tree for first $k$'s then $j$'s possible commitments in \cref{fig:example_G1}. As pictured, $j$ first makes a commitment, which can predicate on $k$'s commitment, but cannot see their move. Next, $k$ can make a commitment, with full knowledge of $j$'s strategy. Agent $j$'s strategies are labeled with their response based on whether $k$ commits to left or right, respectively. That is, in commitment $LL$, agent $j$ commits to playing left regardless of $k$'s commitment and, in $LR$, $j$'s commitments match $k$'s.  Notice that it is always $i$'s local SPE to move right, resulting in $u^3$, the old SPE and the other two's favorite, in $j$'s right three branches. This is not so in $RR$, but it is not clear that $j$ would ever make that commitment, given that the others all appear to offer a better outcome.

Rather than expand the tree for $i$, which results in a very large graph, we instead explore the cuts that $i$ can make. First, note that $k$'s least favorite outcome, $u^4$, is $i$'s best.  This means that no combination of  $i$'s cuts can be used to persuade $k$ to move toward $u^4$.  For example, if $j$ makes the $LR$ commitment, $i$ could cut away $u^3$ in favor of $u^1$ unless $k$ moves right, but the choice for $k$, between $u^1$ and $u^4$ would still see $k$ opting for $u^1$.  Instead, $i$'s best commitment does not take into account $k$'s commitment at all and is as follows:

\begin{quote}
    \begin{enumerate}
        \item[$i$:] ``\emph{Unless $j$ commits to $RR$, I will cut away the right branch.}'' 
    \end{enumerate}
\end{quote}

This means that when $j$ does commit to $RR$, i.e. when $j$ commits to $i$'s desired $u^4$ regardless of what $k$ commits to, $i$ can be guaranteed $u^4$ if they either commit to or play left.  For every $j$ commitment that is not $RR$, $i$ commits to playing left, towards $k$'s node. With this set up, we now go through the backwards induction of the other agents, given the knowledge of $i$'s commitment. Starting with the final agent, $k$, who has the following position:

\begin{quote}
    \begin{enumerate}
        \item[$k$:] ``\emph{Unless $j$ commits to $RR$, agent $i$ will play toward my node, giving me a choice between $u^1$ and $u^2$. Among these I prefer $u^2$. Regardless of whether or not I deploy a commitment, if I go, I will go last, so whichever outcome I commit or play for I will certainly receive. So, unless agent $j$ commits to $RR$ I will opt for $u^2$.}''
    \end{enumerate}
\end{quote}
Agent $j$ can then make the following inference:

\begin{quote}
    \begin{enumerate}
        \item[$j$:] ``\emph{Unless I commit to $RR$, agent $i$ will play toward agent $k$, leaving agent $k$ with only the choice between $u^1$ and $u^2$ and they will certainly pick $u^2$. If I do commit to $RR$, then agent $i$ will play towards my node and I will be obligated to play $u^4$, regardless of what agent $k$ might commit to.  Thus I have a choice between $u^2$ and $u^4$ and will pick $u^4$ and commit to obeying the threat.}''
    \end{enumerate}
\end{quote}

Thus $j$ will commit to $RR$ and the game will result in the same equilibrium as $G^1$.  The interesting observation about this situation is that $i$'s threat does not depend on $k$ making a specific commitment.  Instead, $i$ is counting on the fact that if $k$ were to make a \emph{threat to the threat}, as in the case when $k$ has the first commitment, the commitment to that threat comes after all other agents have made commitments.  With this timing, $k$ knows they certainly will have to follow through with their threat if they make it.  Thus $j$ knows that the help $k$ offered when they had the first commitment is no longer viable and the old threat from $G^1$ still stands.

\begin{definition}[Downward Closure] \label{downward}
Stackelberg resilience has \emph{downward closure} if $k$-resilience implies $\ell$-resilience for all $\ell \leq k$. 
\end{definition}

The property of \emph{downward closure} means that if a game is resilient for some number of commitments, that resilience will still hold with fewer commitments. We will now show that this property holds. As a warm-up, consider the case of 1-resilience from 2-resilience: here, we claim that omitting the second commitment from a 2-resilient game still yields the same equilibrium. If the game in which only agent 1 has a commitment had a different outcome, there must be at least one node in $G$ for which the corresponding set $I$ is different. Let $G^*$ be the lowest such node and observe that by definition this cannot be a leaf. Suppose $G^*$ were owned by 2; given that 2 has no commitment, 2 will pick the local SPE child from $I^L \cup I^R$. It is easy to see that $\textsf{threaten}(I^L,I^R)\subseteq I^L\subseteq L^L$, where $L^L$ denotes the leaves of the left subtree, and analogously so on the right, $\textsf{threaten}(I^R,I^L)\subseteq I^R\subseteq L^R$. Since this game is 2-resilient, the optimal choice for 2 will correspond to both the universal SPE and the 2-commitment choice and therefore also the local SPE. If 2 played their local SPE in the 2-commitment case, any commitment they had with regards to that particular node was trivial so the removal thereof will not affect the potential threats. Given that the $I^L$ and $I^R$ are the same, but $G^*$ is the lowest deviant node, we have a contradiction. If instead, $G^*$ is owned by 1, $I^L\cup I^R$ will be the same as in the 2-commitment case since $I^L$ and $I^R$ are the same. Furthermore, the threaten mechanism is unchanged and 1 may use any inducible leaf to threaten for any leaf below that node, just as before.  Since $I^L$ and $I^R$ are the same by assumption, there are no new threats to make.  Thus $I$ is unchanged and we again have a contradiction. Thus removing the second commitment in a 2-resilient game will not change the equilibrium. Applying this to both orders of commitment arrangement gives the desired result.

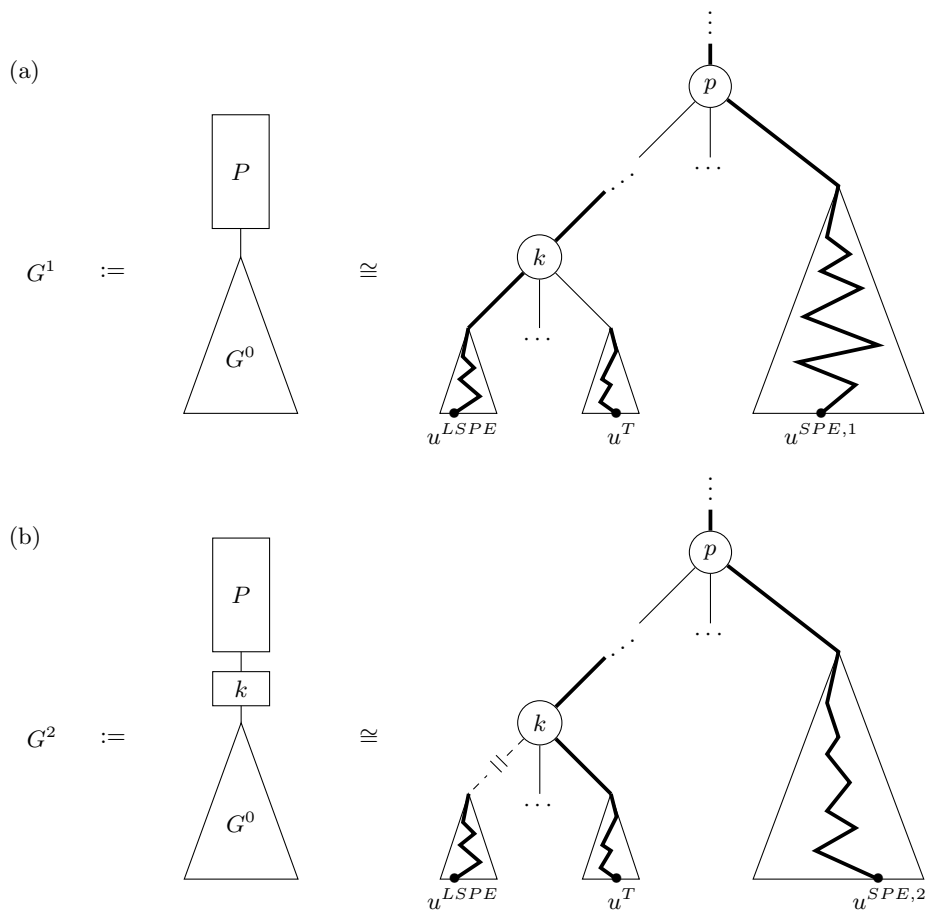
\begin{figure}
    \centering
    \begin{scaletikzpicturetowidth}{\columnwidth}
    \begin{tikzpicture}[scale=\tikzscale]
        \node at (-7.8, 4.5) {(a)};
        \node at (-7.5, 1) {$G^1$};
        \node at (-6.25,0.98) {$:=$};
        \node at (-1.75,1) {$\cong$};
        \draw (-4,1.25) -- (-5,-1.5) -- (-3,-1.5) -- cycle;
        \node at (-4,-0.5) {$G^0$};
        \draw (-4.5, 1.75) -- (-3.5, 1.75) -- (-3.5,3.75) -- (-4.5,3.75) -- cycle;
        \node at (-4,2.75) {$P$};
        \draw (-4,1.25) -- (-4,1.75);
        
        \node at (4.25,5.5) {$\vdots$};
        \node[draw,circle] at (4.25,4.25) (gp) {$p$};
        \draw (gp) -- (4.25,3);
        \node at (4.25,2.8) {$\cdots$};
        \draw[line width=0.5mm] (4.25,5) -- (gp);
        \draw[line width=0.5mm] (gp) -- (6.5,2.5) -- (6.3,1.6) -- (6.7, 1.3) -- (6.2,1) -- (6.9,0.7) -- (5.9,0.2) -- (7.2,-0.3) -- (5.8,-0.6) -- (6.8,-1) -- (6.2,-1.5) node {$\bullet$};
        \node at (6.2,-1.85) {$u^{SPE,1}$};
        \draw (6.5,2.5) -- (5,-1.5) -- (8,-1.5) -- cycle;
        \node[draw,circle] at (1.25,1.25) (gk) {$k$};
        \draw[line width=0.5mm] (gk) -- (0,0);
        \draw[line width=0.5mm] (gk) -- (2.4,2.4);
        \draw (3,3) -- (gp);
        \draw (gk) -- (1.25,0);
        \node at (1.25,-0.2) {$\cdots$};
        \node at (2.7,2.7) {$\iddots$};
        \draw (gk) -- (2.5,0);
        \draw (0,0) -- (-0.5,-1.5) -- (0.5,-1.5) -- cycle;
        \draw[line width=0.5mm] (0,0) -- (-0.1,-0.5) -- (0.1, -0.7) -- (-0.15,-0.9) -- (0.2,-1.2) -- (-0.25,-1.5) node {$\bullet$};
        \begin{scope}[shift={(2.5,0)}]
            \draw (0,0) -- (-0.5,-1.5) -- (0.5,-1.5) -- cycle;
            \draw[line width=0.5mm] (0,0) -- (0.1,-0.4) -- (-0.15, -0.9) -- (0,-1) -- (-0.18,-1.3) -- (0.1,-1.5) node {$\bullet$};
        \end{scope}
        \node at (-0.1,-1.85) {$u^{LSPE}$};
        \node at (2.7,-1.85) {$u^{T}$};
    \end{tikzpicture}\end{scaletikzpicturetowidth}\\
    \begin{scaletikzpicturetowidth}{\columnwidth}
    \begin{tikzpicture}[scale=\tikzscale]
        \node at (-7.8, 4.5) {(b)};
        \node at (-7.5, 1) {$G^2$};
        \node at (-6.25,0.98) {$:=$};
        \node at (-1.75,1) {$\cong$};
        \draw (-4,1.25) -- (-5,-1.5) -- (-3,-1.5) -- cycle;
        \node at (-4,-0.5) {$G^0$};
        \draw (-4.5, 2.5) -- (-3.5, 2.5) -- (-3.5,4.5) -- (-4.5,4.5) -- cycle;
        \node at (-4,3.5) {$P$};
        \node at (-4,1.85) {$k$};
        \draw (-4,1.25) -- (-4,1.55);
        \draw (-4.5,1.55) -- (-3.5,1.55) -- (-3.5,2.15) -- (-4.5,2.15) -- cycle;
        \draw (-4,2.15) -- (-4,2.5);
        
        \node at (4.25,5.5) {$\vdots$};
        \node[draw,circle] at (4.25,4.25) (gp) {$p$};
        \draw (gp) -- (4.25,3);
        \node at (4.25,2.8) {$\cdots$};
        \draw[line width=0.5mm] (4.25,5) -- (gp);
        \draw[line width=0.5mm] (gp) -- (6.5,2.5) -- (6.3,1.6) -- (6.4, 1.3) -- (6.5,1) -- (6.3,0.7) -- (6.7,0.2) -- (6.3,-0.3) -- (6.8,-0.6) -- (6.1,-1) -- (7.2,-1.5) node {$\bullet$};
        \node at (7.4,-1.85) {$u^{SPE,2}$};
        \draw (6.5,2.5) -- (5,-1.5) -- (8,-1.5) -- cycle;
        \node[draw,circle] at (1.25,1.25) (gk) {$k$};
        \draw[dashed] (gk) -- node[fill,white,circle]{$\,$}  (0,0);
        \draw (0.68,0.48) -- (0.48,0.68);
        \draw (0.58,0.38) -- (0.38,0.58);
        \draw[line width=0.5mm] (gk) -- (2.4,2.4);
        \draw (3,3) -- (gp);
        \draw (gk) -- (1.25,0);
        \node at (1.25,-0.2) {$\cdots$};
        \node at (2.7,2.7) {$\iddots$};
        \draw[line width=0.5mm] (gk) -- (2.5,0);
        \draw (0,0) -- (-0.5,-1.5) -- (0.5,-1.5) -- cycle;
        \draw[line width=0.5mm] (0,0) -- (-0.1,-0.5) -- (0.1, -0.7) -- (-0.15,-0.9) -- (0.2,-1.2) -- (-0.25,-1.5) node {$\bullet$};
        \begin{scope}[shift={(2.5,0)}]
            \draw (0,0) -- (-0.5,-1.5) -- (0.5,-1.5) -- cycle;
            \draw[line width=0.5mm] (0,0) -- (0.1,-0.4) -- (-0.15, -0.9) -- (0,-1) -- (-0.18,-1.3) -- (0.1,-1.5) node {$\bullet$};
        \end{scope}
        \node at (-0.1,-1.85) {$u^{LSPE}$};
        \node at (2.7,-1.85) {$u^{T}$};
    \end{tikzpicture}\end{scaletikzpicturetowidth}
    \caption{(a) The game $G^1$ as an extensive-form game. The agent $p \in J$ at some point chooses between the subgame that leads to $u^{SPE,1}$ or the subgame that leads to the key tree, the subgame in which agent $k$ chooses between $u^{LSPE}$ and $u^T$. (b) The game $G^2$ as an extensive-form game. The agent $p \in J$ at some point chooses between the subgame that leads to $u^{SPE,2}$ or the subgame that leads to the subgame where agent $k$ chooses between $u^{LSPE}$ and $u^T$. In order for the SPE to have changed from $u^{SPE,1}$ to $u^{SPE,2}$, agent $k$ must have cut away the subgame that corresponds to $u^{LSPE}$.}
    \label{fig:G1}
\end{figure}

\begin{theorem}\label{thm:downward_trans}
    Stackelberg resilience has downward closure.
\end{theorem}
\begin{proof}
    We show the contrapositive: assume there is game $G^0$ which is not $(k-1)$-resilient, then we claim it is also not $k$-resilient. We assume for simplicity that $G^0$ is in generic form. Let $u^{SPE,0}$ be the utility vector of the subgame perfect equilibrium. We know that $G^0$ is not $(k-1)$-resilient, so there is a list $P$ of $(k-1)$ agents such that giving these agents commitments in the order specified by $P$ results in a utility vector $u^{SPE,1} \neq u^{SPE,0}$. Let $G^1$ be the game that starts with these commitments and ends with $G^0$ (see \cref{fig:G1}). Let $I \subseteq P$ be the agents for whom $u_i^{SPE,1} > u_i^{SPE,0}$, and let $J = P \setminus I$ be its complement. We know that both $I,J \neq \emptyset$, because if $I$ is empty then we would not have $u^{SPE,1} \neq u^{SPE,0}$ and similarly for $J$. 

    Now let $k \in [n] \setminus P$ be arbitrary and define the game $G^2$ that starts with commitments in the order specified by $L$, then has a commitment for agent $k$, and finally a subgame with $G^0$ (see \cref{fig:G1} (b)). Let $u^{SPE,2}$ be the utility vector of the SPE. We can assume w.log. that the only agents in $G^0$ are $I \cup J \cup \{k\}$ (as we can collapse the subgames otherwise).  Our claim is that $u^{SPE,2} \neq u^{SPE,0}$. If $u^{SPE,2} = u^{SPE,1}$ then we are done, so assume $u^{SPE,2} \neq u^{SPE,1}$. Now consider a subgame where the root node is owned by $k$. Define the \emph{local SPE} as the subgame perfect equilibrium for this subgame and call the arrived at utility vector $u^{LSPE}$. In order for $u^{SPE,2} \neq u^{SPE,1}$, there must exist a subgame, the \emph{key tree}, and its corresponding local SPE, such that agent $k$ commits to an action that results in a different utility vector (for that subgame), say $u^T$. By subgame perfection, we must have $u^T_k < u^{LSPE}_k$.

    Since the equilibrium changed from $G^0$ to $G^1$ there must be some subset of agents, say $P \subseteq J$, who committed in $G^1$ to actions that they would not have played in $G^0$, and that these changes resulted in $u^{SPE,1}$. These agents have to have been threatened, as they are now worse off. We know there is a threat from $I$ against $J$ in $G^1$. If this threat were not $u^{LSPE}$ then the equilibrium would not change from $G^1$ to $G^2$ by the introduction of the commitment for $k$ that moves from $u^{LSPE}$ to $u^T$ in that subgame. Let $p \in P$ be any agent who is being threatened using $u^{LSPE}$. Assume for simplicity there is only one such $p$. In order for $u^T$ to nullify the threat of $u^{LSPE}$, we have to have $u^T_p > u^{LSPE}_p$. Thus, the only disruption that $k$ can affect is to commit to $u^T$. Now, consider the following commitment: 
    \begin{quote}
        \begin{enumerate}
            \item[$I$:] ``\emph{I will cut as to necessarily reach the subgame owned by $p$ in every subgame where $p$ does not commit to $u^{SPE,1}$; otherwise I will make the same cut as I did in the commitment that resulted in $u^{SPE,1}$.}'' 
        \end{enumerate}
    \end{quote}
     Consider the choice faced by $k$. Since they are the last agent with a commitment they can see the actions stipulated in the commitments of all the other agents. As argued, agent $k$ can either commit to playing towards $u^{LSPE}$ or they can commit to playing towards $u^T$. If $p$ commits to $u^{SPE,1}$ then we are done. Then we know, based on $I$'s commitment, that $p$ will be faced with the choice of the key tree or $u^{SPE,1}$. Thus if $k$ commits towards $u^T$, they can infer that $p$ will choose the key tree since $u^T_p > u^{SPE,1}_p$. Thus, $k$ sees that a commitment towards $u^T$ will result in that outcome. Consider the choice faced by $p$. If they comply with $I$ and commit to $u^{SPE,1}$ they will get it. If they do not comply with $I$ and play for the key tree, they know that agent $k$ -- seeing this commitment -- will not commit to $u^T$ (as they would prefer $u^{LSPE}$ when given this choice). This gives $p$ the choice between $u^{LSPE}$ and $u^{SPE,1}$, for which the latter outcome is their best response. Thus, we have demonstrated a commitment for $I$ in which they can prevent going back to the old equilibrium. In particular, this means that $G^0$ is not $k$-resilient. \qed
\end{proof}
This can also be viewed as a monotonicity property of the commitments: namely, that any outcome that is inducible with the use of commitments for some agent is still inducible when adding a commitment at the end for another agent.  Note that it is not necessary that the $G^2$ equilibrium is the same as the $G^1$.  The addition of a commitment for $k$ could easily open up even better opportunities for the agents that benefited from the $G^1$ commitments.

\bibliographystyle{alpha}
\bibliography{refs}

\begin{thebibliography}{GDSH12b}

\bibitem[AK19]{contract_1}
Aditya Asgaonkar and Bhaskar Krishnamachari.
\newblock Solving the buyer and seller’s dilemma: A dual-deposit escrow smart contract for provably cheat-proof delivery and payment for a digital good without a trusted mediator.
\newblock In {\em 2019 IEEE International Conference on Blockchain and Cryptocurrency (ICBC)}, pages 262--267, 2019.

\bibitem[CS06]{stackelberg_complexity}
Vincent Conitzer and Tuomas Sandholm.
\newblock Computing the optimal strategy to commit to.
\newblock In {\em Proceedings of the 7th ACM Conference on Electronic Commerce}, EC '06, page 82–90, New York, NY, USA, 2006. Association for Computing Machinery.

\bibitem[GDSH12a]{reverse_stackelberg}
Noortje Groot, Bart De~Schutter, and Hans Hellendoorn.
\newblock Reverse stackelberg games, part i: Basic framework.
\newblock In {\em 2012 IEEE International Conference on Control Applications}, pages 421--426, 2012.

\bibitem[GDSH12b]{reverse_stackelberg_2}
Noortje Groot, Bart De~Schutter, and Hans Hellendoorn.
\newblock Reverse stackelberg games, part ii: Results and open issues.
\newblock In {\em 2012 IEEE International Conference on Control Applications}, pages 427--432. IEEE, 2012.

\bibitem[GDSH14]{groot2014toward}
Noortje Groot, Bart De~Schutter, and Hans Hellendoorn.
\newblock Toward system-optimal routing in traffic networks: A reverse stackelberg game approach.
\newblock {\em IEEE Transactions on Intelligent Transportation Systems}, 16(1):29--40, 2014.

\bibitem[GZDS17]{groot2017hierarchical}
Noortje Groot, Georges Zaccour, and Bart De~Schutter.
\newblock Hierarchical game theory for system-optimal control: Applications of reverse stackelberg games in regulating marketing channels and traffic routing.
\newblock {\em IEEE Control Systems Magazine}, 37(2):129--152, 2017.

\bibitem[HAS21]{smart_contracts}
Mathias Hall-Andersen and Nikolaj~I. Schwartzbach.
\newblock Game theory on the blockchain: A model for games with smart contracts.
\newblock In Ioannis Caragiannis and Kristoffer~Arnsfelt Hansen, editors, {\em Algorithmic Game Theory}, pages 156--170, Cham, 2021. Springer International Publishing.

\bibitem[KCP10]{korzhyk2010complexity}
Dmytro Korzhyk, Vincent Conitzer, and Ronald Parr.
\newblock Complexity of computing optimal stackelberg strategies in security resource allocation games.
\newblock In {\em Proceedings of the AAAI Conference on Artificial Intelligence}, volume~24, pages 805--810, 2010.

\bibitem[LCC84]{inducibleregion}
Peter~B. Luh, Shi-Chung Chang, and Tsu-Shuan Chang.
\newblock Brief paper: Solutions and properties of multi-stage stackelberg games.
\newblock {\em Automatica}, page 251–256, March 1984.

\bibitem[LS23]{landis2023stackelberg}
Daji Landis and Nikolaj Schwartzbach.
\newblock {\em Stackelberg Attacks on Auctions and Blockchain Transaction Fee Mechanisms}.
\newblock 09 2023.

\bibitem[Ric53]{rice}
Henry~Gordon Rice.
\newblock Classes of recursively enumerable sets and their decision problems.
\newblock {\em Transactions of the American Mathematical Society}, 74(2):358--366, 1953.

\bibitem[Sch21]{contract_2}
Nikolaj~I. Schwartzbach.
\newblock An incentive-compatible smart contract for decentralized commerce.
\newblock In {\em 2021 IEEE International Conference on Blockchain and Cryptocurrency (ICBC)}, pages 1--3, 2021.

\bibitem[TKG20]{tajeddini2020decentralized}
Mohammad~Amin Tajeddini, Hamed Kebriaei, and Luigi Glielmo.
\newblock Decentralized hierarchical planning of pevs based on mean-field reverse stackelberg game.
\newblock {\em IEEE Transactions on Automation Science and Engineering}, 17(4):2014--2024, 2020.

\bibitem[vS34]{stackelberg}
Heinrich von Stackelberg.
\newblock {\em Marktform und Gleichgewicht}.
\newblock Verlag von Julius Springer, 1934.

\bibitem[Woo14]{ethereum}
Gavin Wood.
\newblock Ethereum: A secure decentralised generalised transaction ledger.
\newblock {\em Ethereum project yellow paper}, 151:1--32, 2014.

\end{thebibliography}

\end{document}